\documentclass[AMA,STIX1COL]{WileyNJD-v2}

\articletype{Article Type}%

\received{26 April 2016}
\revised{6 June 2016}
\accepted{6 June 2016}

\usepackage{amsmath,amssymb,amsfonts}
\usepackage{graphicx}
\usepackage{amsmath,amssymb,amsfonts}
\usepackage{graphicx}
\usepackage{textcomp}

\usepackage{amsbsy}
\usepackage{fixmath}
\usepackage{bm}

\usepackage{graphics}
\usepackage{epstopdf}

\begin{document}

\title{Nonlinear Controller Design with Prediction Horizon Time Reduction Applied to Unstable CSTR System}
%\protect\thanks{This is an example for title footnote.}}

\author[1]{Chinmay Rajhans*}

\author[2]{Sowmya Gupta}

%\author[3]{Author Three}

\authormark{AUTHOR ONE \textsc{et al}}

\address[1]{\orgdiv{Department of Electrical Engineering}, \orgname{IIT Bombay, Mumbai}, \orgaddress{\state{Maharashtra}, \country{India}}}

\address[2]{\orgdiv{Center for Research in Nano Technology and Science}, \orgname{IIT Bombay, Mumbai}, \orgaddress{\state{Maharashtra}, \country{India}}}

%\address[3]{\orgdiv{Org Division}, \orgname{Org Name}, \orgaddress{\state{State name}, \country{Country name}}}

\corres{*Corresponding author name, \email{rajhanschinmay2@gmail.com}}
%This is sample corresponding address. 

\presentaddress{Research Scholar, Department of Electrical Engineering, IIT Bombay, Powai, Mumbai 400076 Maharashtra, India}

\abstract[Summary]{Ensuring nominal asymptotic stability of the Nonlinear Model Predictive Control controller is not trivial. Stabilizing ingredients such as terminal penalty term and terminal region are crucial in establishing the asymptotic stability. Current work presents alternate approaches namely arbitrary controller based approach and linear quadratic regulator based approach, which provide larger degrees of freedom for enlarging the terminal region as against conservative approaches from the literature. Efficacy of the proposed approaches is demonstrated using benchmark two state continuous stirrer tank reactor system around an unstable operating point. Terminal regions obtained using the arbitrary controller based approach and linear quadratic regulator based approach are approximately 45 and 412 times larger by area measure when compared to the largest terminal region obtained using the approach from the literature. As a result, there is significant reduction in the prediction and control horizon time.  }

%Both the approaches are scalable to system of any dimension. Approach from the literature provides a scalar whereas proposed approaches provide two additive matrices as tuning parameters for shaping of the terminal region. Proposed approaches involve solving modified Lyapunov equations to compute terminal penalty term, followed by explicit characterization of the terminal region. 


\jnlcitation{\cname{%
\author{Williams K.}, 
\author{B. Hoskins}, 
\author{R. Lee}, 
\author{G. Masato}, and 
\author{T. Woollings}} (\cyear{2016}), 
\ctitle{A regime analysis of Atlantic winter jet variability applied to evaluate HadGEM3-GC2}, \cjournal{Q.J.R. Meteorol. Soc.}, \cvol{2017;00:1--6}.}

\maketitle

%\footnotetext{\textbf{Abbreviations:} ANA, anti-nuclear antibodies; APC, antigen-presenting cells; IRF, interferon regulatory factor}
\footnotetext{\textbf{Abbreviations:} LQR, Linear Quadratic Regulator, NMPC, Nonlinear Model Predictive Control; TR, Terminal Region; ODEs, Ordinary Differential Equations}

\section{Introduction}
\label{sec:introduction}
Stability and performance are two crucial factors to be taken into consideration while designing any controller. One of the most promising optimization based controller is the Model Predictive Control (MPC). MPC finds applications in every field of science, engineering and technology \cite{Mayne1990, Qin2003, Camacho2007}. % Two crucial aspects for any controller design are stability and performance. 
Various researchers have presented overview of MPC schemes \cite{Allgower1999, Mayne2000}. Primary concept for ensuring nominal stability involves inclusion of stabilizing constraints \cite{Rawlings1993, Michalska1993}. Commonly used stabilizing constraints include a) Terminal equality constraint, b) Terminal penalty term, and c) Terminal inequality constraint \cite{Allgower1999, Fontes2001, Rawlings2009}. A significant process has taken place in the area of nominal stability of linear MPC \cite{Muske1993} and Nonlinear MPC (NMPC) \cite{Oliveira1994, Sistu1996, Mayne2014, Rawlings2017}. Grimm et al. have presented examples when a significantly small change in any of the model parameters can alter the stability characteristics of NMPC \cite{Grimm2004}. Hence formally establishing the asymptotic stability becomes important, necessary and challenging.

Executing terminal equality constraint is convenient \cite{Keerthi1988}, however main limitation is that it is highly conservative and often leads to infeasibility specifically when using constrained formulations. Michalska and Mayne conceptualized dual mode MPC scheme where in the idea of terminal region was introduced \cite{Michalska1993}. NMPC controller is expected to drive the plant trajectory into a region, termed as terminal region, around the set point in a finite time using the feasible inputs. Subsequently local linear controller will take the system trajectory to the set point. This idea was extended by Chen and Allg\"ower where in NMPC controller was used inside the terminal region instead of using a linear controller, which has resulted in the concept of Quasi Infinite Horizon - Nonlinear Model Predictive Control (QIH-NMPC) scheme \cite{Chen1998}.

Region of attraction for NMPC is a set of initial conditions which result in all the constraints being satisfied with feasible inputs within a specified finite time. It may be noted that the size of the terminal region is directly correlated to the size of the feasible region i.e. the region of attraction. For a given finite horizon formulation with constant prediction horizon time, larger the terminal region results in a larger region of attraction. Having larger region of attraction indicates ability of controller to converge to the desired operating point from an initial condition which is far away from the set point \cite{Mayne2000}. Alternately, for identical initial conditions, controller would require smaller prediction horizon time to satisfy the terminal inequality constraint. Mhaskar et al. presented asymptotically stable NMPC design for continuous time switched systems \cite{Mhaskar2005}. Major drawbacks are explicit characterization of the feasible initial conditions and applicability to only switched systems.

Limon et al. presented design of NMPC without terminal inequality constraint. Concept involved appropriate scaling the terminal penalty term to compensate for the difference due to absence of the terminal inequality constraint \cite{Limon2006}. It may be noted that there is a limitation as to what extent designer can increase the terminal penalty term and results in smaller region of attraction. Pannocchia et al. presented an algorithm to convert infinite horizon constrained linear quadratic regulator formulation into a finite dimensions quadratic programming problem after assuming piece-wise linear inputs \cite{Pannocchia2010}. However the issue of convergence of solution and sub-optimality need to be addressed. Esterhuizen et al. presented NMPC asymptotic stability results without stabilizing terminal ingredients. However, two key assumptions of sufficiently longer prediction horizon and cost controllability assumption limit the applicability of the algorithm to limited systems \cite{Esterhuizen2021}. Jadbabaie et al. present unconstrained NMPC stability results without terminal ingredients. The approach makes use of gradual reduction of Lyapunov function eventually resulting in an asymptotic stability characteristics \cite{Jadbabaie2001}. However, amount of time required to reach the desired operating points may be very large and also the design is suitable for unconstrained systems. The proposed approach in this work is suitable for any kind of nonlinear continuous time system with inputs constraints.

Chen and Allg\"ower presented an approach for the computation of the terminal penalty term and also for the characterization of the terminal region for the continuous time NMPC formulation \cite{Chen1998}. Research involves local linearization at the set point followed by solving a modified Lyapunov equation. Subsequently Chen and Allg\"ower provide an approach to numerically characterize the terminal region using an inequality based conditions. First major drawback of their approach is a tuning parameter which is nearly independent of the NMPC formulation stage weighting matrices. Second limitation of Chen and Allg\"ower's approach is that, it provides a single scalar tuning parameter which restricts the design to one degree of freedom for shaping of the terminal region, hence, resulting in a very conservative terminal region. Chen and Allg\"ower's approach makes use of Linear Quadratic Regulator (LQR) controller, which is designed using the stage cost weighting matrices and in turn do not provide any degrees of freedom to the controller designer.

Several researchers have developed approaches for the terminal region characterization for NMPC formulations for the discrete time cases \cite{Limon2002, Johansen2004, Rajhans2017, Yu2017, Rajhans2019}. It may be noted that discrete time formulations require a separate considerations due to the concept of sampling time vastly affecting the terminal region shape and size \cite{Astrom1997, Grune2011, Rajhans2019}. Although the approaches developed for the discrete time QIH-NMPC formulations provide large degrees of freedom for enlarging the terminal region, however, their application to continuous time QIH-NMPC formulations is very limited. Hence, there is a need to develop approaches for the terminal region characterization for the continuous time NMPC formulations which provide large degrees of freedom.

Chen and Allg\"ower \cite{Chen1998a} established that terminal inequality constraint can be avoided when the terminal penalty term and the prediction horizon is chosen sufficiently large for continuous time NMPC formulation. However, the result is applicable only for stable set points or stable continuous time nonlinear systems. In general, for any kind of system, without the terminal ingredients, nominal stability of NMPC controller is not guaranteed. In addition, when the terminal inequality constraint is avoided, typically designer is required to use a relatively larger prediction horizon time which increases the computational burden significantly. Such limitation can be overcome by using terminal inequality constraint which assists in reducing the prediction horizon time \cite{Chen1998}.

The approach by Chen and Allg\"ower \cite{Chen1998} is based on linear controller designed at the origin and is applicable to any continuous time nonlinear system governed by Ordinary Differential Equations (ODEs). Lucia et al. \cite{Lucia2015} have extended this work by making use of nonlinear controller for design of the terminal ingredients. Their approaches is based on Taylor series expansions of the system dynamics with considering higher order terms of the stage weighting matrices. However, this approach is applicable to only a special class of continuous time systems where in time derivatives of the system dynamics are polynomial functions. In this work, two approaches are presented which are applicable to any type of nonlinear continuous time system governed by ODEs.

Rajhans et al. presented alternate arbitrary controller based approach for the computation of the terminal penalty and for the characterization of the terminal region for the continuous time NMPC formulations \cite{Rajhans2016}. Arbitrary controller based approach makes use of a single additive matrix as the tuning parameter for shaping of the terminal region. Current work converts norm based method to inequality based method, which assists in enlarging the terminal region. In the proposed approached in the current work, two tuning matrices are provided which further increase the degrees of freedom available with the controller designer. Proposed approach provides three degrees of freedom namely a) linear stabilizing controller, b) additive state weighting matrix, and c) additive input weighting matrix. Current work proposed one novel LQR based approach for the terminal region characterization, which provides two additive weighting matrices for enlarging the terminal region. 
 
%Additionally, nominal asymptotic stability result of the continuous time QIH-NMPC formulation is presented after incorporating the updated terminal ingredients obtained after solving the modified Lyapunov equation. 

Efficacy of the proposed approaches with three tuning parameters is demonstrated using simulations on a benchmark Chemical engineering system called Continuous Stirrer Tank Reactor (CSTR) \cite{Hicks1971}.  Various researchers have used two state CSTR system for demonstrating their controller performance \cite{Tenny2004, Ghaffari2013, Ellis2014, Narasingam2019, Ramesh2021}. However, application of the continuous time quasi infinite horizon NMPC with guaranteed stability is very limited and one additional novelty of the current work. In the demonstration example, it can be observed that the proposed approaches result in significantly larger terminal regions when compared to the approaches available in the literature. Work also presents closed loop simulations of the system under continuous time NMPC controller to validate the applicability of the controller in practical scenarios. Results pertaining to the reduction of the prediction horizon time are presented in detail.

Second section presents the continuous time NMPC formulation in detail. In addition, approach by Chen and Allg\"ower \cite{Chen1998} is stated formally along with its limitation. Third section presents the proposed arbitrary controller based approach using inequality method for the computation of the terminal penalty and for the characterization of the terminal region. In addition, third section also presented novel LQR based approach for the terminal region characterization. Subsequently, asymptotic stability result is presented. Forth section presents numerical characterization of the terminal region using the approaches presented in the third section.  Fifth section presents the terminal region characterization using demonstration case study. Sixth section details the CSTR continuous time simulation and results obtained using the CSTR case study. Seventh section gives the conclusions from the theory and cases study.

\section{Continuous Time NMPC Formulation} 
Consider a continuous time nonlinear system is given by 
\begin{eqnarray}
\frac{d{\mathbf{X}(t)}}{dt} = {\mathbf{f}_c} ({\mathbf{X}}(t),{\mathbf{U}}(t)) \label{csystem1} 
\end{eqnarray}
where $\mathbf{X}(t) \in \mathbb{R}^{n_x}$ denotes the state
vector in absolute terms and $\mathbf{u}(t) \in \mathbb{R}^{n_u}$ denotes the input vector in absolute terms. Let $(\mathbf{X}_s, \mathbf{U}_s)$ be the constant steady state of the system (\ref{csystem1}) i.e. $\mathbf{0} = {\mathbf{f}_c} (\mathbf{X}_s, \mathbf{U}_s)$. Defining shift of origin as follows: 
\begin{eqnarray}
\mathbf{x}(t) = \mathbf{X}(t) - \mathbf{X}_s \\ 
\mathbf{u}(t) = \mathbf{U}(t) - \mathbf{U}_s
\end{eqnarray}
After shift of origin, consider the continuous time nonlinear system given as  
\begin{eqnarray}
\frac{d({\mathbf{X}(t)-\mathbf{X}_s})}{dt} = {\mathbf{f}_c} ({\mathbf{X}(t)-\mathbf{X}_s},{\mathbf{U}(t)-\mathbf{U}_s}) \label{csystem2} 
\end{eqnarray}
Rewiring using a simpler notation gives 
\begin{align}
\frac{d{\mathbf{x}(t)}}{dt} &= {\mathbf{f}} ({\mathbf{x}}(t),{\mathbf{u}}(t)) \label{csystem} \\
{\bf{x}}(0) &= {\bf{x}}_0
\end{align}
where $\mathbf{x}(t) \in \mathcal{X} \subset \mathbb{R}^{n_x}$ denotes the state
vector and $\mathbf{u}(t) \in \mathcal{U} \subset \mathbb{R}^{n_u}$ denotes the input vector.

Assumptions are stated as follows: 
\begin{description}
\item[C1] System dynamics function $\mathbf{f}: \mathbb{R}^{n_x}\times \mathbb{R}^{n_u} \to \mathbb{%
R}^{n_x}$ is twice continuously differentiable. 
\item[C2] The origin $\mathbf{0} \in \mathbb{R}^{n_x}$ is an equilibrium point of the system (\ref{csystem}) i.e. $\mathbf{f}\left( \mathbf{0}, \mathbf{0} \right) =\mathbf{0}$.
\item[C3] The inputs $\mathbf{u}(t)$ are constrained inside 
a closed and convex set $\mathcal{U} \subset \mathbb{R}^{n_u}$.
\item[C4] The system (\ref{csystem}) has a unique solution for any initial
condition $\mathbf{x}_{0} \in \mathcal{X}$ and any piece wise right continuous input $\mathbf{u}%
(\cdot):[0,\infty) \to $ $\mathcal{U}$. 
\item[C5] The state $\mathbf{x}(t)$ is perfectly known at any time $t$ i.e. all the states are measured. 
\item[C6] External disturbances do not affect the system dynamics. 
\end{description}

\subsection{NMPC Formulation}
For the continuous time system given by (\ref{csystem}), NMPC formulation is stated as follows: 

\begin{equation}
\begin{array}{c}
\min \\ 
\overline{\mathbf{u}}_{[t, t+T_{p}]}%
\end{array}%
J\left( \mathbf{x}(t),\overline{\mathbf{u}}_{[t, t+T_{p}]}\right)  \label{COptimal}
\end{equation}%
with 
\begin{eqnarray}
J\left( \mathbf{x}(t),\overline{\mathbf{u}}_{[t, t+T_{p}]}\right)
&=&\int_{t}^{t+T_{p}}\left\{ 
%\begin{array}{c}
\mathbf{z}(\tau)^T \mathbf{W}_{x} \mathbf{z}(\tau) + \overline{\mathbf{u}}(\tau)^T \mathbf{W}_{u} \overline{\mathbf{u}}(\tau) 
%\end{array}%
\right\} d\tau  +\mathbf{z}(t+T_p)^T \mathbf{P} \mathbf{z}(t+T_p) \label{StageCost} 
\end{eqnarray}%
\begin{equation}
\overline{\mathbf{u}}_{[t, t+T_{p}]}=\left\{ \mathbf{u}(\tau )\in \mathcal{U}%
:\tau \in \left[t, t+T_{p}\right] \right\} \label{InputSet}
\end{equation}%
subject to 
\begin{eqnarray} 
\frac{d\mathbf{z}(\tau)}{d\tau }&=&\mathbf{f}\left( \mathbf{z}(\tau ), %
\overline{\mathbf{u}}(\tau )\right) \text{ for } \tau \in \left[ t, t+T_{p}\right] \label{PredictedState} \\ 
\mathbf{z}(t)&=&\mathbf{x}(t) \label{InitialCondition} \\ 
\mathbf{z}\left( t+T_{p}\right) &\in& \Omega  \label{TerminalRegion}
\end{eqnarray}%
where $\mathbf{W}_{x}$ and $\mathbf{W}_{u}$ are state and input weighting matrices of dimension $\left( n_x \times n_x \right) $, $\left( n_u \times n_u \right) $ respectively. $\mathbf{P}$ is the terminal penalty matrix of dimension $\left( n_x \times n_x \right) $. $\mathbf{W}_{x}, \mathbf{W}_{u}, \mathbf{P}$ are symmetric positive definite matrices. $T_{p}$ is a finite prediction horizon time and is identical to the control horizon time. 
%Here, $\left\Vert \mathbf{v}\right\Vert _{\mathbf{A}}^{2}=\mathbf{v}^{T}\mathbf{Av}$. 
$\mathbf{z}(\tau )$ denotes the predicted state in the NMPC formulation and $%
\overline{\mathbf{u}}(\tau )$ denotes the future control input moves. The
set $\Omega$ is termed as the \emph{terminal region} in the neighborhood of
the origin. 
% and is chosen such that it is invariant for the nonlinear control system controlled by a \emph{fictitious} local linear state feedback controller with gain matrix say, $\mathbf{K}$. 
The set $\mathcal{X}_{T_p} \subset \mathcal{X} \subset \mathbb{R}^{n_x}$ is termed as the \emph{region of attraction} is the set of all feasible initial conditions i.e. it is a set of all initial conditions $\mathbf{x}_0$ such that terminal inequality constraint (\ref{TerminalRegion}) is satisfied with inputs constrained given by equation (\ref{InputSet}) satisfied. 

\subsection{Design and Implementation of NMPC Formulation}
The terminal region $\Omega$ is chosen as an invariant set for the nonlinear system (\ref{csystem}) controlled by local linear controller with gain matrix $\mathbf{K}$. The terminal penalty term is chosen such that for all trajectories starting from any point inside the terminal region $\Omega$, with approximation that a single cost term having larger value that the sum of all the predicted stage cost terms from end of horizon to infinity and is given as follows:  
\begin{equation}
\mathbf{z}(t+T_p)^T \mathbf{P} \mathbf{z}(t+T_p) \geq \int_{t+T_{p}}^{\infty }\left\{ 
%\begin{array}{c}
\mathbf{z}(\tau)^T \mathbf{W}_{x} \mathbf{z}(\tau)
+ \overline{\mathbf{u}}(\tau) \mathbf{W}_{u} \overline{\mathbf{u}}(\tau)
%\end{array} 
\right\} d\tau  \label{TRCondition}
\end{equation}
with $\overline{\mathbf{u}}(\tau) = -\mathbf{K} \mathbf{z}(\tau) \in \mathcal{U}$ for all $\tau \geq (t+T_p)$ and for all $\mathbf{z}(t+T_p) \in \Omega$. 

It is assumed that the solution to the optimal problem (\ref{COptimal}) with stage cost defined by (\ref{StageCost}) with input set given by (\ref{InputSet}) subject to the predicted state dynamics (\ref{PredictedState}) with initial condition (\ref{InitialCondition}) and terminal constraint (\ref{TerminalRegion}) i.e. $\overline{\mathbf{u}}_{[t, t+T_{p}]}^*$ exists and can be computed numerically. Controller is implemented as a moving horizon framework. Accordingly, only the first control move 
\begin{eqnarray}
\mathbf{u}(t) = \overline{\mathbf{u}}^*(t) \label{InputMove}
\end{eqnarray}
is implemented in the plant. Entire process is repeated at next time point $t+\delta$ with $\delta$ being an sufficiently small sampling period. The term \emph{Quasi Infinite} is because of the fact that the NMPC formulation deplicts the stability properties of the infinite horizon formulation, however, the actual implementation is finite horizon. Such implementation is achieved with the help of the equation (\ref{TRCondition}). However, only ensuring the terminal penalty term satisfying the condition (\ref{TRCondition}) is not sufficient to guarantee the nominal asymptotic stability of the NMPC controller, hence terminal constraint as given by (\ref{TerminalRegion}) becomes inevitable. It may be noted that local linear controller with gain matrix $\mathbf{K}$ is not used for implementation of the NMPC controller and is only a mathematical construct to characterize the terminal region $\Omega$.

\subsection{Chen and Allg\"ower's Approach}
Before proceeding to the proposed arbitrary controller based approach, a look at Chen and Allg\"ower's approach is required. Consider, Jacobian linearization of the nonlinear system (\ref{csystem}) in the neighborhood the origin as,  
\begin{equation}
\frac{d\mathbf{x}(t)}{dt}=\mathbf{Ax}(t)+\mathbf{Bu}(t)  \label{CLinSys}
\end{equation}%
where 
\begin{equation*}
\mathbf{A=}\left[ \frac{\partial \mathbf{f}}{\partial \mathbf{x}}\right]
_{\left( \mathbf{0},\mathbf{0}\right) }\text{ and \ }%
\mathbf{B=}\left[ \frac{\partial \mathbf{f}}{\partial \mathbf{u}}\right]
_{\left( \mathbf{0},\mathbf{0}\right) }
\end{equation*}

One additional assumption is required at this stage. 
\begin{description}
\item[C7] The linearized system (\ref{CLinSys}) is stabilizable. 
\end{description}
Chen and Allg\"ower characterize the terminal region as, 
\begin{equation}
\Omega \equiv \left\{ \mathbf{x} \in \mathbb{R}^{n} | \mathbf{x}^{T}%
\mathbf{P} \mathbf{x} \leq \alpha, -\mathbf{Kx} \in \mathcal{U}%
\right\}
\end{equation}%
where linear gain $\mathbf{K}$ and the terminal penalty matrix $\mathbf{P}$ are the steady state solutions of the modified Lyapunov equation given as follows:    
\begin{equation}
\left( \mathbf{A}_{K} + \kappa \mathbf{I}\right) ^{T} \mathbf{P} +%
\mathbf{P} \left( \mathbf{A}_{K}+\kappa \mathbf{I}\right) =-\mathbf{Q%
}^*  \label{ChenLyapunov}
\end{equation}%
\begin{equation}
\mathbf{Q}^* = \mathbf{W}_{x} + \mathbf{K}^T \mathbf{W}_{u} \mathbf{K}  \label{Qstar}
\end{equation}
where $\mathbf{A}_{K} = \mathbf{A-BK}$ and parameter $\kappa > 0$ is chosen such that $\kappa < -Re \left[ \lambda _{\max }\left( \mathbf{A-BK} \right) \right]$. Note $Re \left[ \lambda _{\max }\left( \mathbf{A-BK} \right) \right]$ is the real part of the right most eigenvalue of $\mathbf{A}_{K}$ i.e. eigen value having largest real part and it is negative due to the fact that linear matrix $\mathbf{A}_{K}$ is stable by design. It can be noted that once stage cost weighting matrices $\mathbf{W}_{x}, \mathbf{W}_{u}$ are chosen, there is barely any degree of freedom left to the designer for shaping of the terminal region. This results in a very conservative terminal regions. The limitation is overcome by using the arbitrary controller based approach wherein additive tuning matrices are introduced which provide large degrees of freedom for enlarging of the terminal region and is presented in the subsequent section. 

\section{Alternate Approaches for the Terminal Region Characterization} 
In the arbitrary controller based approach, an arbitrary stabilizing linear controller is designed using any of the methods available in the literature such as pole placement \cite{Kailath1980, Albertos2006}, linear quadratic Gaussian control \cite{Kirk1970} and so on. 
We prove the following lemma for the arbitrary controller based approach: 

\begin{lemma} \label{lemma1} Suppose that assumptions C1 to C7 are satisfied and a stabilizing
feedback control law is designed i.e. $\mathbf{A}_{K}=(%
\mathbf{A-BK})$ is stable indicating all the eigenvalues have negative real part. Let $\Delta \mathbf{Q}$ is any positive definite matrix. Let matrix $\mathbf{P}$ denote the solution of the following modified Lyapunov equation: 
\begin{equation}
\mathbf{A}_{K}^T \mathbf{P} + \mathbf{P} \mathbf{A}_{K}=-(%
\mathbf{Q}^* + \Delta \mathbf{Q)}  \label{ACLyap}
\end{equation}%
where $\mathbf{Q}^*$ is defined by equation (\ref{Qstar}). Then there exists a
constant $\alpha > 0$ which defines an ellipsoid of the form 
\begin{equation}
\Omega \equiv \left\{ \mathbf{x} \in \mathbb{R}^{n_x} | \mathbf{x}^{T} \mathbf{%
P} \mathbf{x} \leq \alpha , -\mathbf{Kx} \in \mathcal{U} \right\}
\label{ACTR}
\end{equation}%
such that $\Omega$ is an invariant set for the nonlinear system given by (\ref{csystem}) with linear controller $\mathbf{u}(t) = - \mathbf{Kx}(t)$. Additionally, for any $\mathbf{x}(t+T_p) \in \Omega$ the inequality given by (\ref{TRCondition1}) holds true. 
\begin{equation}
\mathbf{z}(t+T_p)^T \mathbf{P} \mathbf{z}(t+T_p) \geq \int_{t+T_{p}}^{\infty }\left\{ 
\begin{array}{c}
\mathbf{z}(\tau)^T \mathbf{W}_{x} \mathbf{z}(\tau)
+ \overline{\mathbf{u}}(\tau) \mathbf{W}_{u} \overline{\mathbf{u}}(\tau)
\end{array} \right\} d\tau  \label{TRCondition1}
\end{equation}
\end{lemma}

\begin{proof}
Since $\mathbf{A}_{K}=(\mathbf{A-BK})$ is stable, hence, the eigenvalues of $\mathbf{A}_{K}$ are having negative real part. Using the solvability condition of the modified Lyapunov equation, a unique $\mathbf{P} > 0$ can be computed which solves the equation (\ref{ACLyap}). According to Assumption C2, the origin $\mathbf{0} \in \mathbb{R}^{n_u}$ is in the interior of the input constraints set $\mathcal{U}$.
Accordingly, we can compute a constant $\gamma$ which defined a set $\Omega_{\gamma}$ 
such that 
\begin{equation}
\Omega_\gamma \equiv \left\{ \mathbf{x} \in \mathbb{R}^{n_x} | \mathbf{x}^{T}%
\mathbf{P} \mathbf{x} \leq \gamma, - \mathbf{Kx} \in \mathcal{U}%
\right\} \label{Omegagamma}
\end{equation}%
Now, let $0 < \alpha \leq \gamma$ specify a region of the form given by
equation (\ref{ACTR1}). 
\begin{equation}
\Omega \equiv \left\{ \mathbf{x} \in \mathbb{R}^{n_x} | \mathbf{x}^{T} \mathbf{%
P} \mathbf{x} \leq \alpha \right\}
\label{ACTR1}
\end{equation}%
As the input constraints are satisfied in $\Omega
_\gamma$ and $\Omega \subseteq \Omega_\gamma$ (by virtue of $0 < \alpha \leq \gamma$), the system dynamics can be equivalently viewed as an input unconstrained system in the set $\Omega$. 
Consider a vector $\mathbf{\Phi}_{K}(\mathbf{x})$ representing the nonlinearity in the system dynamics defined as 
\begin{equation}
\mathbf{\Phi }_{K}(\mathbf{x})=\mathbf{f}(\mathbf{x, -Kx}) - \mathbf{A}_{K}%
\mathbf{x}  \label{PhyK} 
\end{equation}%
Note for a linear system $\mathbf{\Phi}_{K}(\mathbf{x}) = \mathbf{0}$. 
Consider a Lyapunov candidate defined as 
\begin{equation}
V(\mathbf{x}) = \mathbf{x}^{T} \mathbf{P} \mathbf{x}  \label{Vx}
\end{equation}%
The time derivative of $V(\mathbf{x})$ can be expressed as follows: 
\begin{align}
\frac{dV(\mathbf{x})}{dt}& = \frac{d \mathbf{x}^{T}}{dt} \mathbf{P} \mathbf{x} + \mathbf{x}^{T} \mathbf{P} \frac{d \mathbf{x}}{dt}
\label{Vdot1}
\end{align}%
Substituting from (\ref{PhyK}) into (\ref{Vdot1}), 
\begin{align}
\frac{dV(\mathbf{x})}{dt} = \mathbf{x}^{T} \left( \mathbf{A}_{K}^{T} \mathbf{P} + \mathbf{P} \mathbf{A}_{K} \right) \mathbf{x} + 2 \mathbf{x}^{T} \mathbf{P} \mathbf{\Phi }_{K} \mathbf{(x)}
\label{Vdot2}
\end{align}%
Using equation (\ref{ACLyap}) into (\ref{Vdot2}), 
\begin{align}
\frac{dV(\mathbf{x})}{dt} = -\mathbf{x}^{T} \left( \mathbf{Q}^*+ \Delta \mathbf{Q} \right) \mathbf{x} +2\mathbf{x}^{T} \mathbf{P} \mathbf{\Phi }_{K} \mathbf{(x)} \label{Vdot3}
\end{align}%
Rearranging results in the following equation: 
\begin{align}
\frac{dV(\mathbf{x})}{dt} = -\mathbf{x}^{T} \mathbf{Q}^* \mathbf{x} + \left( -\mathbf{x}^{T} \Delta \mathbf{Q} \mathbf{x} +2\mathbf{x}^{T} \mathbf{P} \mathbf{\Phi }_{K} \mathbf{(x)} \right) \label{Vdot3b}
\end{align}%

There are two possibility to characterize the terminal region. First is a norm based method and second is the inequality based method. \\ 
Method A - Norm based method: 
Taking norm of second term of the equation (\ref{Vdot3}), 
\begin{align}
\mathbf{x}^{T} \mathbf{P} \mathbf{\Phi }_{K} \mathbf{(x)} \leq |P| L_\Phi |\mathbf{x}|^2 \label{Vdot4}
\end{align}%
Since $\mathbf{x}^T \Delta \mathbf{Q} \mathbf{x} \ge \lambda_{min}(\Delta \mathbf{Q})$ and combining (\ref{Vdot4}) into (\ref{Vdot3}), 
\begin{align}
\frac{dV(\mathbf{x})}{dt} \le -\mathbf{x}^{T} \mathbf{Q}^* \mathbf{x} - \left[ \lambda_{min}(\Delta \mathbf{Q}) - 2 |\mathbf{P}| L_\Phi \right] |\mathbf{x}|^2 \label{Vdot5}
\end{align}%
If $\Omega$ is chosen such that 
\begin{align}
\left[ \lambda_{min}(\Delta \mathbf{Q}) - 2 |\mathbf{P}| L_\Phi \right] \leq 0 \label{Vdot6}
\end{align}%
then 
\begin{align}
\frac{dV(\mathbf{x})}{dt} \le -\mathbf{x}^{T} \mathbf{Q}^* \mathbf{x} \label{Vdot7}
\end{align}%

Method B - Inequality based method: 
Rearranging terms from the equation (\ref{Vdot3}), 
\begin{align}
\frac{dV(\mathbf{x})}{dt} = -\mathbf{x}^{T} \mathbf{Q}^* \mathbf{x} + \left( -\mathbf{x}^{T} \Delta \mathbf{Q} \mathbf{x} +2\mathbf{x}^{T} \mathbf{P} \mathbf{\Phi }_{K} \mathbf{(x)} \right)
\label{Vdot11}
\end{align}%
Consider second term of the expression (\ref{Vdot11}), 
\begin{equation}
\mathbf{\Psi} (\mathbf{x}) := \left(  \mathbf{x}^{T} \Delta \mathbf{Q} \mathbf{x} -2 \mathbf{x}^{T} \mathbf{P} \mathbf{\Phi }_{K} \mathbf{(x)} \right) \label{PsiDef}
\end{equation}
Using (\ref{PsiDef}) in (\ref{Vdot11}), 
\begin{align}
\frac{dV(\mathbf{x})}{dt} = -\mathbf{x}^{T} \mathbf{Q}^* \mathbf{x} - \mathbf{\Psi} (\mathbf{x})
\label{Vdot12}
\end{align}%
If $\Omega$ is chosen such that 
\begin{align}
\mathbf{\Psi} (\mathbf{x}) = \left(  \mathbf{x}^{T} \Delta \mathbf{Q} \mathbf{x} -2 \mathbf{x}^{T} \mathbf{P} \mathbf{\Phi }_{K} \mathbf{(x)} \right) \geq 0 \label{Vdot13}
\end{align}%
then 
\begin{align}
\frac{dV(\mathbf{x})}{dt} \le -\mathbf{x}^{T} \mathbf{Q}^* \mathbf{x} \label{Vdot14}
\end{align}%
Equation (\ref{Vdot14}) for inequality based method is identical to equation (\ref{Vdot7}) for norm based method.

Integrating inequality (\ref{Vdot7}) or (\ref{Vdot14}) over the interval, $[t+T_{p}, \infty
),$ it follows that 
\begin{equation}
V(\mathbf{x(}t+T_{p}))\geq \int_{t+T_{p}}^{\infty} \mathbf{x}(\tau)^{T}%
\mathbf{Q}^* \mathbf{x}(\tau) d\tau
\end{equation}%
i.e. inequality (\ref{TRCondition1}) holds true for any $\mathbf{x(}t+T_{p}) \in \Omega$.
\end{proof}

\begin{lemma} \label{lemma2} Suppose that assumptions C1 to C7 are satisfied. Let $\widetilde{\mathbf{W}}_{x} > \mathbf{W}_x$ and $\widetilde{\mathbf{W}}_{u} > \mathbf{W}_u$ be any positive definite matrices. Let matrix $\mathbf{P}_{LQ}$ denote the solution of the following modified Lyapunov equations: 
\begin{equation}
\mathbf{A}_{{K}_{LQ}}^{T}\mathbf{P}_{LQ}+\mathbf{P}_{LQ}\mathbf{A}_{{K}_{LQ}} =-\left( \widetilde{\mathbf{W}}_{x}+\mathbf{K}%
_{LQ}^{T}\widetilde{\mathbf{W}}_{u}\mathbf{K}_{LQ}\right) \label{CARE}
\end{equation}
\begin{equation}
\mathbf{K}_{LQ}\mathbf{=}\left( \widetilde{\mathbf{W}}_{u}\right) ^{-1}%
\mathbf{B}^{T}\mathbf{P}_{LQ}  \label{Kgain}
\end{equation}
%where $\mathbf{Q}^*$ is defined by equation (\ref{Qstar}). 
where $\mathbf{A}_{{K}_{LQ}} = \mathbf{A}-\mathbf{B} \mathbf{K}_{LQ}$. Then there exists a
constant $\alpha > 0$ which defines an ellipsoid of the form 
\begin{equation}
\Omega \equiv \left\{ \mathbf{x} \in \mathbb{R}^{n_x} | \mathbf{x}^{T} \mathbf{%
P} \mathbf{x} \leq \alpha , -\mathbf{K}_{LQ} \mathbf{x} \in \mathcal{U} \right\}
\label{LQRTR}
\end{equation}%
such that $\Omega$ is an invariant set for the nonlinear system given by (\ref{csystem}) with linear controller $\mathbf{u}(t) = - \mathbf{K}_{LQ} \mathbf{x}(t)$. Additionally, for any $\mathbf{x}(t+T_p) \in \Omega$ the inequality given by (\ref{TRCondition1}) holds true. 
\begin{equation}
\mathbf{z}(t+T_p)^T \mathbf{P} \mathbf{z}(t+T_p) \geq \int_{t+T_{p}}^{\infty }\left\{ 
\begin{array}{c}
\mathbf{z}(\tau)^T \mathbf{W}_{x} \mathbf{z}(\tau)
+ \overline{\mathbf{u}}(\tau) \mathbf{W}_{u} \overline{\mathbf{u}}(\tau)
\end{array} \right\} d\tau  \label{TRCondition1}
\end{equation}
\end{lemma}

\begin{proof}
Proof is similar to the proof of Lemma \ref{lemma1} except for minor changes such as $\mathbf{K}$ is replaced by $\mathbf{K}_LQ$, $\mathbf{P}$ is replaced by $\mathbf{P}_LQ$ and remaining changes are shown below: 
Consider a candidate Lyapunov function defined as 
\begin{equation*}
V(\mathbf{x})=\mathbf{x}^{T}\mathbf{P}_{LQ}\mathbf{x}
\end{equation*}
%where $\mathbf{P}_{LQ}$ is the solution of equation (\ref{CARE}). 
Using equation (\ref{CARE}), the time derivative of $V(\mathbf{x})
$ can be expressed as follows 
\begin{equation}
\frac{dV(\mathbf{x})}{dt}=\mathbf{x}^{T}\left( \mathbf{A}_{K}^{T}\mathbf{P}_{LQ}%
\mathbf{+P}_{LQ}\mathbf{A}_{K}\right) \mathbf{x}+2\mathbf{x}^{T}\mathbf{P}_{LQ}%
\mathbf{\phi (x)}  \label{Vder}
\end{equation}%
Defining matrices, 
\begin{equation}
\mathbf{\Delta W}_{x}\equiv \widetilde{\mathbf{W}}_{x}-\mathbf{W}_{x}>0\text{
and }\mathbf{\Delta W}_{u}\equiv \widetilde{\mathbf{W}}_{u}-\mathbf{W}_{u}>0
\label{deltaW}
\end{equation}%
one can write 
\begin{equation}
\widetilde{\mathbf{W}}_{x}+\mathbf{K}_{LQ}^{T}\widetilde{\mathbf{W}}_{u}%
\mathbf{K}_{LQ}\mathbf{=Q}^{\ast }+\mathbf{\Delta Q}
\end{equation}%
\begin{eqnarray}
\mathbf{Q}^{\ast } &=&\mathbf{W}_{x}+\mathbf{K}_{LQ}^{T}\mathbf{W}_{u}%
\mathbf{K}_{LQ} \\
\mathbf{\Delta Q} &\mathbf{=}&\mathbf{\Delta W}_{x}+\mathbf{K}%
_{LQ}^{T}\Delta \mathbf{W}_{u}\mathbf{K}_{LQ}
\end{eqnarray}%
and the equation (\ref{CARE}) can be re-written as follows 
\begin{equation}
\mathbf{A}_{K}{}^{T}\mathbf{P}_{LQ}+\mathbf{P}_{LQ}\mathbf{A}_{K}=-\left( 
\mathbf{Q}^{\ast }+\Delta \mathbf{Q}\right)   \label{CLyap1}
\end{equation}%
Equation (\ref{Vder}) and equation (\ref{CLyap1}) are combined as follows: 
\begin{equation}
\frac{dV(\mathbf{x})}{dt}=-\mathbf{x}^{T}\mathbf{(\mathbf{Q}^{\ast }+\Delta 
\mathbf{Q})x}+2\mathbf{x}^{T}\mathbf{P}_{LQ}\mathbf{\phi (x)}  \label{Teq}
\end{equation}%
Rearranging results in the following equation: 
\begin{equation}
\frac{dV(\mathbf{x})}{dt}=-\mathbf{x}^{T}\mathbf{Q}^{\ast} \mathbf{x} + \left( -\mathbf{x}^{T} \Delta 
\mathbf{Q} \mathbf{x}+2\mathbf{x}^{T}\mathbf{P}_{LQ}\mathbf{\phi (x)} \right)  \label{Teq2}
\end{equation}%
Equation (\ref{Teq2}) is identical to the equation (\ref{Vdot3b}). Rest of the proof is similar to the proof of Lemma \ref{lemma1}. Both the methods i.e. norm based method and inequality based method are applicable for the LQR based approach as well.  

\end{proof}
Consider the feasibility lemma as follows: 
\begin{lemma} \label{lemma3} Let the assumptions C1-C7 hold true. For the nominal continuous time system, feasibility of continuous time QIH-NMPC formulation problem (\ref{COptimal}) at time $t=0$ implies its feasibility for all $t > 0$.  
\end{lemma}
\begin{proof}
Proof is identical to the proof of the lemma 2 from \cite{Chen1998}. 
\end{proof}
Consider the asymptotic stability result as follows: 
\begin{theorem} \label{theorem1} Let a) Assumptions C1-C7 hold true and b) the continuous time NMPC problem is feasible at $t = 0$. The nominal nonlinear system (\ref{csystem}) controlled with NMPC controller is asymptotically stable at the origin. 
\end{theorem}

\begin{proof}
From equation (\ref{Vx}) from the lemma \ref{lemma1} or \ref{lemma2}, 
consider the Lyapunov candidate function 
\begin{equation}
V(\mathbf{x}) = \mathbf{x}^{T} \mathbf{P} \mathbf{x} \label{Vx1}
\end{equation}%
Consider the following three properties \cite{Khalil2002}: 
\begin{itemize}
\item $V(\mathbf{0}) = (\mathbf{0}^{T}) \mathbf{P} (\mathbf{0}) = 0$. 
\item Since $\mathbf{P}$ is a positive definite matrix, $V(\mathbf{x}) = \mathbf{x}^{T} \mathbf{P} \mathbf{x} > 0$ for all $\mathbf{x} \neq \mathbf{0}$. 
\item Using (\ref{Vdot7}) or (\ref{Vdot14}) and $\mathbf{Q}^* > 0$ implies 
\begin{align}
\frac{dV(\mathbf{x})}{dt} \le -\mathbf{x}^{T} \mathbf{Q}^* \mathbf{x} < 0 \label{Vx2}
\end{align}%
\end{itemize}

Thus, the candidate function $V(\mathbf{x})$ is a Lyapunov function for the
nonlinear system for $\mathbf{x} \in \Omega$ under NMPC controller. 
Hence, the closed loop system is asymptotically stable at the origin. 
%Detailed steps of the proof are similar to the proof of the theorem 1 from \cite{Chen1998}. 
\end{proof}
Note $\mathbf{K}$ is to be read as $\mathbf{K}_{LQ}$ for linear gain matrix and $\mathbf{P}$ is to be read as $\mathbf{P}_{LQ}$ for terminal penalty matrix for the subsequent sections for the application of LQR based approach. Notation is simplified for readability. 

\section{Terminal Region Characterization} 
Lemma \ref{lemma1} or Lemma \ref{lemma2} gave conditions for explicit characterization of the terminal region. It is possible to numerically compute the terminal region and subsequently implement the QIH-NMPC controller. 
% Let $\mathbf{u} = u_1, u_2, ..., u_{n_u}$ where $u_{i} \in \mathbb{R}$ for $i = 1, 2, ..., n_u$.
\subsection{Steps for the Characterization of the Terminal Region}
Steps for characterization of the terminal region using arbitrary controller based approach are given below: 

\begin{description}
\item[S1] Computation of Upper Bound Set: \\ 
Compute the largest value of $\gamma$ such that inputs constraints are satisfied in the set $\Omega_\gamma$. 
\begin{equation}
\Omega_\gamma \equiv \left\{ \mathbf{x} \in \mathbb{R}^{n_x} | \mathbf{x}^{T}%
\mathbf{P} \mathbf{x} \leq \gamma, - \mathbf{Kx} \in \mathcal{U}%
\right\} \label{Omegagamma1}
\end{equation}%
This can be formulated as a simple Quadratic Programming (QP) problem if the constraints are defined by upper bound and lower bound on each of the input signal. Typically the set $\Omega_\gamma$ would be tangential to at least one of the input constraint. 
\item[S2a] Computation of the Terminal Region using norm based method: \\ 
Compute the largest $\alpha \in (0, \gamma]$ such that 
\begin{align}
L_\Phi \leq L_\Phi^* = \frac{\lambda_{min}(\Delta \mathbf{Q})}{2 |\mathbf{P}|} \label{TRCompute1}
\end{align}%
where 
\begin{align}
L_\Phi = \begin{array}{c} 
\max \\ 
\mathbf{x} \in \Omega%
\end{array}%
\frac{|\mathbf{\Phi}_K (\mathbf{x})|}{|\mathbf{x}|} \label{TRCompute2}
\end{align}
This is identical to the method given by Rajhans et al. in \cite{Rajhans2016} for the arbitrary controller based approach. 
\item[S2b] Computation of the Terminal Region using inequality based method: \\ 
Compute the largest $\alpha \in (0, \gamma]$ such that 
\begin{align}
\left[ 
\begin{array}{c} 
\min \\ 
\mathbf{x}(k) \in \Omega%
\end{array}%
\mathbf{\Psi}(\mathbf{x}) \right] = 0 \label{TRCompute3}
\end{align}
The condition given by (\ref{TRCompute3}) ensures that $\mathbf{\Psi}(\mathbf{x}) > 0$ for all $\mathbf{x} \in \Omega$, which is the necessary condition to further establish the nominal asymptotic stability. 
\end{description} 
It may be noted that the steps S1 and S2a results in a conservative terminal region and steps S1 and S2b result in larger terminal region. The step S2b is implemented as follows: \\
Initially $\alpha = \gamma$ and condition (\ref{Vdot13}) i.e. $(\mathbf{\Psi}(\mathbf{x}) \geq 0)$ is checked. If (\ref{Vdot13}) is true, then $\alpha = \gamma$. If (\ref{Vdot13}) is false i.e. $(\mathbf{\Psi}(\mathbf{x}) < 0)$ for at least one $\mathbf{x} \in \Omega$, then the value of $\alpha$ is further reduced by a multiplicative factor $\beta < 1$ and $\beta \approx 1$. The process continues until condition (\ref{Vdot13}) is satisfied.

%\subsection{Quantification of the Terminal Region}
Terminal region shape changes according to the computed $\mathbf{P}$ matrix and its size changes according to the value of $\alpha$. In order to compare the size of the terminal regions, area is computed for state dimension of $2$ as 
%, volume for state dimension of 3 and hyper volume for state dimension larger than 3 can be used.   
%The hyper volume of the region enclosed in $\mathbf{x}^{T} \mathbf{P} \mathbf{x} \leq \alpha$ is given
%as
%\begin{align}
%H_V = \frac{c_{n_x} \alpha^{n_x / 2}}{\sqrt{det(\mathbf{P})}} \label{HyperVolume}
%\end{align}
%where $c_{n_x}$ is a constant depending solely on the state dimension $n_x$ and is given by (\ref{HyperVolumeC}).  
%\begin{align}
%c_{n_x} = \frac{\pi^{n_x / 2}}{\Gamma \left( \frac{n_x}{2} + 1 \right) } \label{HyperVolumeC}
%\end{align}
%where $\Gamma$ is a gamma function \cite{Sommerville1929, Rajhans2019}. 
%Initial few values of $c_{n_x}$ include $c_2 = \pi$, $c_3 = 4 \pi /3$, $c_4 = \pi^2 / 2$ and so on. Area for state dimension of 2 is given by 
\begin{align}
A_2 = \frac{\pi \alpha}{\sqrt{det(\mathbf{P})}} \label{Area}
\end{align}
%Volume for state dimension of 3 is given by 
%\begin{align}
%V_3 = \frac{4 \pi \alpha^{3 / 2}}{3 \sqrt{det(\mathbf{P})}} \label{Volume}
%\end{align}
%Hyper volume for state dimension of 4 is given by 
%\begin{align}
%H_{V4} = \frac{\pi^2 \alpha^2}{2 \sqrt{det(\mathbf{P})}} \label{HyperVolume4}
%\end{align}

\section{CSTR Case Study} 
Effectiveness of the proposed approaches for the terminal region characterization and its applicability to NMPC continue time simulations is demonstrated using the benchmark CSTR case study. 
%\begin{itemize}
%\item Three state continuously operated fermenter system \cite{Henson1992} 
%\item Four state quarter car active suspension system \cite{Bououden2016} 
%\end{itemize} 

\subsection{Choice of Tuning Matrices}
According to the design of the arbitrary controller based approach, the gain matrix $\mathbf{K}$ can be any arbitrary stabilizing linear controller. However, in order to simply the computations, simulation results are presented with the following choice. Controller gain $\mathbf{K}$ is the steady state solution of the simultaneous equations (\ref{LQRP}) and (\ref{LQRK}).  
\begin{equation}
\mathbf{A}^T \mathbf{P} + \mathbf{P} \mathbf{A} = - \mathbf{W}_x + \mathbf{P} \mathbf{B} \left( \mathbf{W}_u \right)^{-1} \mathbf{B}^T \mathbf{P} \label{LQRP} 
%= -\mathbf{W}_{x} + \mathbf{P} \mathbf{BK} 
\end{equation}%
\begin{equation}
\mathbf{K} = \left( \mathbf{W}_{u}\right) ^{-1} \mathbf{B}^{T}\mathbf{P}  \label{LQRK}
\end{equation}%

The tuning matrix $\Delta \mathbf{Q}$ is any positive definite matrix. However, in order to simply and structure the computations of the terminal region, following parameterization is carried out:   
\begin{equation}
\Delta \mathbf{Q} = \widetilde{\mathbf{W}}_{x} + \mathbf{K}^T \widetilde{\mathbf{W}}_{u} \mathbf{K}  \label{ACDQ}
\end{equation}%
In order to further simplify the numerical computation of the terminal region, additional parameterization is carried out as follows: 
\begin{equation}
\widetilde{\mathbf{W}}_{x} = \rho_x \mathbf{W}_{x} \text{ and } \widetilde{\mathbf{W}}_{u} = \rho_u \mathbf{W}_{u} \label{TuningM}
\end{equation}
Note that it is sufficient to have $\widetilde{\mathbf{W}}_{x} > \mathbf{W}_{x}$ or $\widetilde{\mathbf{W}}_{u} > \mathbf{W}_{u}$ to satisfy $\Delta \mathbf{Q} > 0$, however, usually both $\widetilde{\mathbf{W}}_{x} > \mathbf{W}_{x}$ and $\widetilde{\mathbf{W}}_{u} > \mathbf{W}_{u}$ is preferred in practice.  
Using the matrices (\ref{TuningM}) into (\ref{ACDQ}), 
\begin{equation}
\Delta \mathbf{Q} = \rho_x \mathbf{W}_{x} + \rho_u \mathbf{K}^T \mathbf{W}_{u} \mathbf{K}  \label{ACDQ1}
\end{equation}
where $\rho_x > 0$ and $\rho_u > 0$ are the tuning scalars. 
Rajhans et al. presented terminal region characterization with only single tuning parameter $\rho_x > 0$ \cite{Rajhans2016}. However, in the current work, two parameters $\rho_x > 0$ and $\rho_u > 0$ are varied for obtaining the terminal region. 

Chen and Allg\"ower presents both the norm based method and inequality based method \cite{Chen1998}. It is reported that inequality based method results in larger terminal region when compared to the norm based method. Approach by Rajhans et al. in \cite{Rajhans2016} makes use of norm based method, however, the proposed approach in this work makes use of inequality based method which is inherently less conservative.

Efficacy of having two tuning parameters is efficiently demonstrated using the case study in the next sub-sections. In the case study, table \ref{CSTR_Iter} steps in which parameters are varied in order to obtain a significantly larger terminal regions. For approach by Chen and Allg\"ower's \cite{Chen1998}, there is a single constant scalar tuning parameter $\kappa$. In the case of arbitrary controller based approach. there are two iterations. In the first iterations, tuning parameter $\rho_x$ is varied keeping $\rho_u$ constant. In the second iteration, value of $\rho_x = \rho_x^*$ where $\rho_x^*$ is the value of $\rho_x$ resulting in maximum terminal region area in the first iteration. Arbitrary controller based approach with single tuning parameter $\rho_x$ is given by Rajhans et al. in \cite{Rajhans2016}.

\begin{table}[tbph]
\caption{Terminal Region Computation Iteration Steps}
\label{CSTR_Iter}\centering%
\begin{tabular}{|l|c|c|c|c|c|}
\hline
$\text{Approach}$ & Iteration & $\begin{array}{c}
\text{Tuning} \\ \text{parameters} 
\end{array}$ & $\begin{array}{c}
\text{Constant} \\ \text{parameters} 
\end{array}$ & $\begin{array}{c}
\text{Initial} \\ \text{value} 
\end{array}$ 
 & $\begin{array}{c}
\text{Increasing} \\ \text{parameter} 
\end{array}$ \\ \hline
Chen and Allg\"ower's \cite{Chen1998} & 1 & $\kappa$ & $\kappa$ & $-0.95*\left[ \lambda _{\max }\left( \mathbf{A-BK} \right) \right]$ & - \\ \hline 
Arbitrary controller based \cite{Rajhans2016} & 1 & $\rho_x$ & $\rho_{u} = 0$ & $\rho_x = 0.1$ & $\rho_x$ \\ \hline 
Arbitrary controller based & 2 & $\rho_x, \rho_u$ & $\rho_{x}^*$ & $\rho_u = 0.1$ & $\rho_u$ \\ \hline 
LQR based & 1 & $\rho_x$ & $\rho_{u} = 1$ & $\rho_x = 1.1$ & $\rho_x$ \\ \hline 
LQR based & 2 & $\rho_x, \rho_u$ & $\rho_{x}^*$ & $\rho_u = 1.1$ & $\rho_u$ \\ \hline
\end{tabular}%
\end{table}

%in the first iteration, parameter $\rho_x$ is increased from $0$ in arbitrary controller approach and from $1$ in LQR based approach to a large value, keeping $\rho_u$ constant at $0$. A value of $\rho_x$ is chosen which leads to larger terminal region in the first iteration. Subsequently, in the second iteration, parameter $\rho_u$ is increased from $0$ to larger value, keeping $\rho_x$ as constant to the values obtained in the first iteration. 

\subsection{CSTR System Details}
Consider Continuous Stirred Tank Reactor (CSTR) initially given by Hicks and Ray\ \cite%
{Hicks1971} and later used by Huang et al. \cite{Huang2012}. The system dynamics equations are: 
\begin{align}
\frac{dz_{c}}{dt}& =\frac{(1-z_{c})}{m_{2}}-k_{0}z_{c} e^{(-E_{a}/z_{T})}
\label{cstrmodel1} \\
\frac{dz_{T}}{dt}=& \frac{(z_{T}^{f}-z_{T})}{m_{2}}%
+k_{0}z_{c} e^{(-E_{a}/z_{T})}-\alpha _{0}m_{1}(z_{T}-z_{T}^{CW})
\label{cstrmodel2}
\end{align}%
where $z_{c}$ and $z_{T}$ represent dimensionless concentration and
dimensionless temperature, respectively. Control inputs are cooling
water flow rate $m_{1}$ and inverse of the dilution rate $m_{2}$. 
%Continuous time Ordinary Differential Equations (ODEs) are converted to discrete time using Runge Kutta order 4 integration method with the help of MATLAB function `ode45' with inputs assumed piece wise constant \cite{Astrom1997}. 

\subsection{Nominal Parameters and Linearization}
Nominal values of the parameters are given in the Table \ref{CSTR_parameters}. 
\begin{table}[tbph]
\caption{CSTR System: Nominal Parameters}
\label{CSTR_parameters}\centering%
\begin{tabular}{|c|c|}
\hline
Variable & Nominal Value \\ \hline\hline
$z_{T}^{CW}$ & $0.38$ \\ \hline
$z_{T}^{f}$ & $0.395$ \\ \hline
$E_{a}$ & $5$ \\ \hline
$\alpha _{0}$ & $1.95\times 10^{-4}$ \\ \hline
$k_{0}$ & $300$ \\ \hline
\end{tabular}%
\end{table}
To improve numerical stability of the optimization routine, the inputs $(m_{1},m_{2})$ appearing in the system dynamics are scaled as $u_{1}=m_{1}/600$ and $u_{2}=m_{2}/40$. Operating point is given as, 
\begin{equation}
\pmb{X}_{s}=\left[\begin{array}{c} 0.6416 \\ 0.5387 \end{array} \right] 
\end{equation}
\begin{equation}
\pmb{U}%
_{s}=\left[ \begin{array}{c} 0.5833 \\ 0.5000 \end{array} \right]  \label{cstrdiscss}
\end{equation}%
The input constraints are given as follows: 
\begin{align}
\mathcal{U}=\left\{u_{1},u_{2}\in \mathbb{R} |  -0.4167\leq u_{1}\leq
0.4167, -0.4750\leq u_{2}\leq 0.5  \right\}
\end{align}%
Jacobian linearization of the continuous time nonlinear system at $(\mathbf{X}_{s},\mathbf{U}_{s})$ yields: 
\begin{equation}
\mathbf{A}=%
\begin{bmatrix}
   -0.0779 &  -0.3088 \\ 
    0.0279 &   0.1905 
\end{bmatrix}%
%\end{equation}
~\text{and }
%\begin{equation}
\mathbf{B}=%
\begin{bmatrix}
         0 &  -0.0358 \\ 
   -0.0184 &   0.0144 
\end{bmatrix}%
\end{equation}%
Eigenvalues of the open loop continuous time dynamics are $(-0.0406, ~0.1532)$, which is unstable (i.e. negative real part).
\subsection{NMPC Controller Design}
Stage cost matrices for the MPC formulation are given as follows: 
\begin{equation}
\pmb{W}_{x}=\left[ 
\begin{array}{cc}
10 & 0 \\ 
0 & 2%
\end{array}%
\right] 
\end{equation}
%\text{ and }
\begin{equation}
\pmb{W}_{u}=\left[ 
\begin{array}{cc}
1 & 0 \\ 
0 & 0.5%
\end{array}%
\right] 
\end{equation}%
Since concentration of the mixture is more crucial compared to the temperature of the reactor, hence, the weight for the first state (concentration) is chosen $5$ times larger when compared to weight of the second state (temperature).  
Sampling interval of $T=1~unit$ is used.

\subsection{Comparison of the Terminal Regions for CSTR System}
Linear gain matrix and terminal penalty matrix obtained using Chen and Allg\"ower's \cite{Chen1998} approach ($\kappa = 0.1059$) is given 
as follows: 
\begin{align}
\mathbf{K}_{CA}=%
\begin{bmatrix}
   -1.6118 &  -10.7187 \\ 
   -2.1094 &  10.5029 
\end{bmatrix}%
,~ \mathbf{P}_{CA}=10^{3}\times 
\begin{bmatrix}
    8.4569 &   5.8384 \\ 
    5.8384 &   4.8968 
\end{bmatrix}
\end{align}
Linear gain matrix and terminal penalty matrix obtained using Arbitrary
Controller based approach ($\rho _{x}=50, \rho _{u}=20$) is given
as follows: 
\begin{align}
\mathbf{K}=%
\begin{bmatrix}
   -1.6118 &  -10.7187 \\ 
   -2.1094 &  10.5029 
\end{bmatrix}%
,~ \mathbf{P}=10^{4}\times 
\begin{bmatrix}
    0.3492 &    0.3406 \\ 
    0.3406 &   1.2265 
\end{bmatrix}
\label{Kgain6}
\end{align}
Linear gain matrix and terminal penalty matrix obtained using LQR based
approach ($\rho _{x}=50,\rho _{u}=1500$) are given as follows: 
\begin{align}
\mathbf{L}_{LQ}=%
\begin{bmatrix}
   -1.2963 & -10.4475 \\ 
    1.1335 &  11.3084 
\end{bmatrix}%
,~ \mathbf{P}_{LQ}=10^{5}\times 
\begin{bmatrix}
    0.1877 &   1.0578 \\ 
    1.0578 &   8.5254 
\end{bmatrix}
\label{Kgain5}
\end{align}%
%
%Figure \ref{fig:Hicks_D_TR_AC_1_LQR_1} compares the terminal regions
%obtained using LQR based approach and Arbitrary Controller based approach.
%It can be observed that the shapes of the terminal regions\ are different
%according to the terminal penalty matrices. Depending on application or
%requirement of the process, one can suitably choose largest terminal region
%given by the two approaches.
%
%\FRAME{ftbphFU}{6.2914in}{3.2461in}{0pt}{\Qcb{CSTR\ system by Hicks and Ray:
%Comparison of largest $\Omega $ \ for $T=1$ obtained using arbitrary
%controller based approach and LQR based approach }}{\Qlb{%
%fig:Hicks_D_TR_AC_1_LQR_1}}{hicks2s_disc_gen2_tr_lqr_1_ac_1.eps}{\special%
%{language "Scientific Word";type "GRAPHIC";maintain-aspect-ratio
%TRUE;display "USEDEF";valid_file "F";width 6.2914in;height 3.2461in;depth
%0pt;original-width 15.0028in;original-height 7.7063in;cropleft "0";croptop
%"1";cropright "1";cropbottom "0";filename
%'../Thesis_Report/Chapter_5/Hicks2S_Disc_Gen2_TR_LQR_1_AC_1.eps';file-properties "XNPEU";}%
%}

%Figure \ref{fig:Hicks_D_TR_AC_1_LQR_1} the largest terminal regions obtained
%with the approaches given in the Subsections \ref{Sec4p3p1} and \ref%
%{Sec4p3p2}. The maximum terminal region area obtained using the LQR based
%approach is approximately $2.8$ times larger than the maximum terminal
%region area obtained using the arbitrary controller based approach.

Table \ref{CSTR_TR_CA_AC_LQR} compares areas of the largest terminal
regions obtained using Chen and Allg\"ower's \cite{Chen1998} (as CA), Arbitrary Controller (as AC) based approach and LQR based approach (as LQ). It can be observed that the terminal region obtained using arbitrary controller based approach is approximately 45 times larger than the area of the terminal region obtained using the approach by Chen and Allg\"ower's \cite{Chen1998}. Additionally, the terminal region obtained using LQR based approach is approximately 412 times and 9 times larger than the area of the terminal region obtained using the approach by Chen and Allg\"ower's \cite{Chen1998} and arbitrary controller based approach respectively.  It can be observed that arbitrary controller based approach using two tuning parameters $\rho_x, \rho_u$ result approximately $4.1$ times increase in area of the terminal region when compared to the arbitrary controller based approach using a single tuning parameter $\rho_x$ as given in \cite{Rajhans2016}. 

%\begin{table}[tbph]
%\caption{CSTR system: comparison of maximum terminal
%regions }
%\label{tbl:CSTR_TR_AC_LQR_3}\centering%
%\begin{tabular}{|c|c|c|c|}
%\hline
%$\text{Approach}$ & $%
%\begin{array}{c}
%\text{Degrees of} \\ \text{freedom} 
%\end{array}%
%$ & $\alpha $ & $\text{Area of }\Omega $ \\ \hline
%AC & $\rho _{x}=50$ & $1.6779$ & $3.781\times 10^{-3}$ \\ 
%\hline
%AC & $\rho _{x}=50$, $\rho _{u}=35$ & $23.0408$ & $%
%8.541\times 10^{-3}$ \\ \hline
%LQR & $\rho _{x}=50$ & $0.4535$ & $1.712\times 10^{-3}$ \\ \hline
%LQR & $\rho _{x}=50$, $\rho _{u}=200$ & $173.0460$ & $2.444\times 10^{-2}$
%\\ \hline
%\end{tabular}%
%\end{table}

\begin{table}[tbph]
\caption{CSTR system: comparison of maximum terminal
regions }
\label{CSTR_TR_CA_AC_LQR}\centering%
\begin{tabular}{|l|c|c|c|c|}
\hline
$\text{Approach}$ & $%
%\begin{array}{c}
%\text{Degrees of} \\ \text{freedom} 
\text{Degrees of freedom} 
%\end{array}%
$ & $\gamma$ & $\alpha $ & $\text{Area of }\Omega $ \\ \hline
Chen and Allg\"ower's \cite{Chen1998} & $\kappa = 0.1059$ & $1.3620$ & $0.1282$ & $%
1.4880 \times 10^{-4}$ \\ \hline
Arbitrary controller based \cite{Rajhans2016} & $\rho _{x}=50, \rho _{u}=0$ & $1.3563$ & $0.6467$ & $0.0016$ \\ \hline 
Arbitrary controller based & $\rho _{x}^*=50, \rho _{u}=20$ & $11.9270$ & $11.9270$ & $0.0067$ \\ \hline 
LQR based & $\rho _{x}=50,\rho _{u}=1$ & $0.0940$ & $0.0435$ & $2.225 \times 10^{-4}$ \\ \hline
LQR based & $\rho _{x}^*=50,\rho _{u}=1500$ & $1.3560 \times 10^{3}$ & $1.3560 \times 10^{3}$ & 0.0614 \\ \hline
\end{tabular}%
\end{table}

\section{NMPC Demonstration Results}
In order to formally demonstrate the efficacy of the larger terminal regions on the MPC, continuous time simulations are carried out using the largest terminal region which is obtained using the novel LQR based approach with two tuning parameters $\rho_x, \rho_u$. Three initial conditions given in the deviation variables and computed in different directions to affirm that the result is certain and not by chance, are given as follows: 
\begin{equation}
\mathbf{x}%
_{P_{1}}(0)= \left[ \begin{array}{c} 
-0.001 \\  -0.050
\end{array} \right]
\text{,~}
\mathbf{x}%
_{P_{2}}(0)= \left[ \begin{array}{c} 
-0.625 \\  0.380
\end{array} \right]
\text{,~}
\mathbf{x}%
_{P_{3}}(0)= \left[ \begin{array}{c} 
0.400 \\  0.230
\end{array} \right]
%\text{,~}
%\mathbf{x}%
%_{P_{4}}(0)= \left[ \begin{array}{c} 
%3 \\ 
%-2 \\ 
%-0.2 \\ 
%0.5
%\end{array} \right] 
\label{Hicks2S_IC1} 
\end{equation} 
Note, in the actual variable terms, the initial conditions for the system become 
\begin{equation}
\mathbf{X}_{Pi}(0) = \mathbf{X}_s + \mathbf{x}_{Pi}(0) \text{ for } i = 1, 2, 3  \label{Hicks2S_IC2} 
\end{equation}

Figure \ref{Hicks2S_Cont_NMPC_result1_actual_states} displays plot of states in actual variables for the MPC simulation. It can be observed that all the states converge to the steady state operating point. 
%States $X_3 = H_3$ and $X_4 = H_4$ are slightly sluggish when compared to the states $X_1 = H_1$ and $X_2 = H_2$. 

\begin{figure}[!ht]
\centerline{\includegraphics[width=\columnwidth]{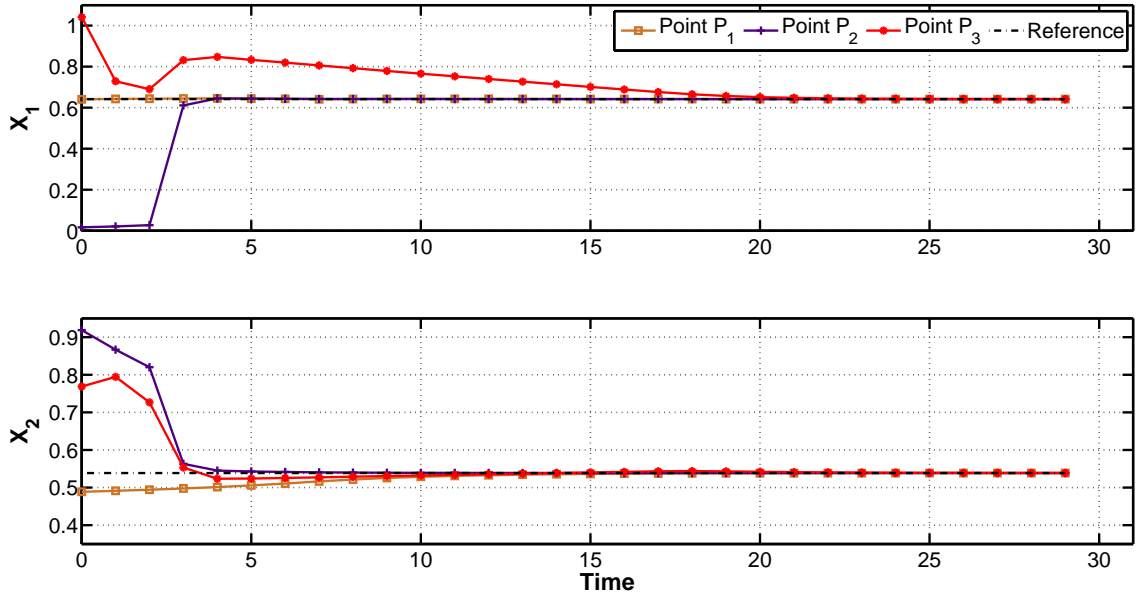}}
\caption{CSTR System: Plot of states in actual variables} 
\label{Hicks2S_Cont_NMPC_result1_actual_states}
\end{figure}

Figure \ref{Hicks2S_Cont_NMPC_result1_deviation_states} shows trajectories of the states in the deviation variables for the MPC simulation. It can be observed that all the states converge to the origin. 

\begin{figure}[!ht]
\centerline{\includegraphics[width=\columnwidth]{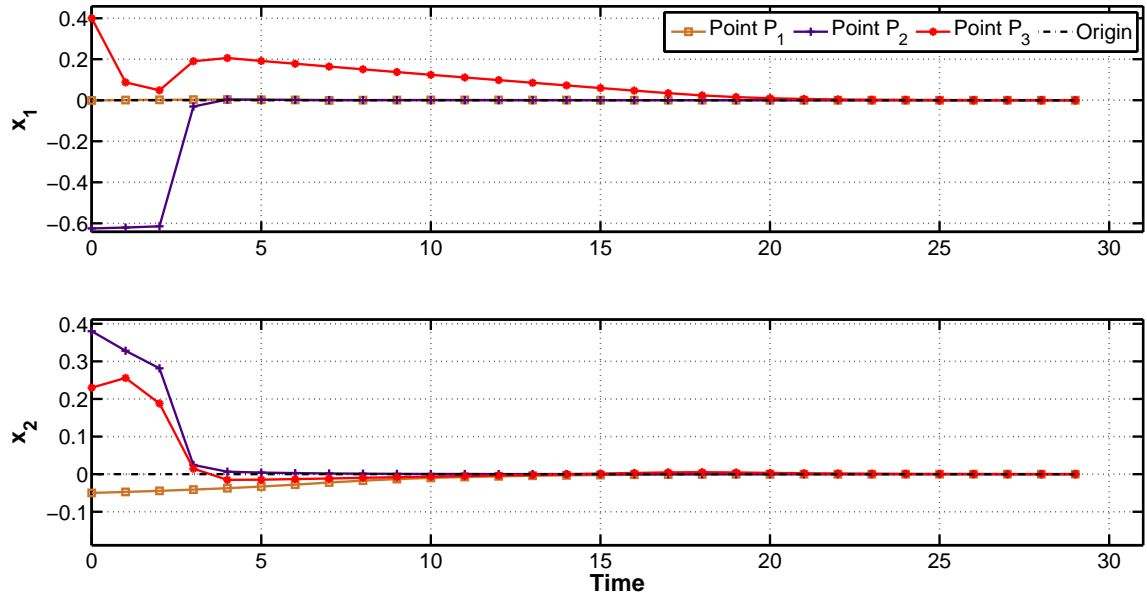}}
\caption{CSTR System: Plot of states in deviation variables} 
\label{Hicks2S_Cont_NMPC_result1_deviation_states}
\end{figure}

Figure \ref{Hicks2S_Cont_NMPC_result1_inputs} shows the plot of the control inputs (as voltage in V). It can be seen that both the control inputs remained inside the limits indicating the feasibility. 
%In addition, one of the input is has started from saturation level indicating that MPC controller attempts to make maximum use of the available control input at the beginning. 
Both the control inputs converge to the  steady state value after sufficient time has elapsed. 
\begin{figure}[!ht]
\centerline{\includegraphics[width=\columnwidth]{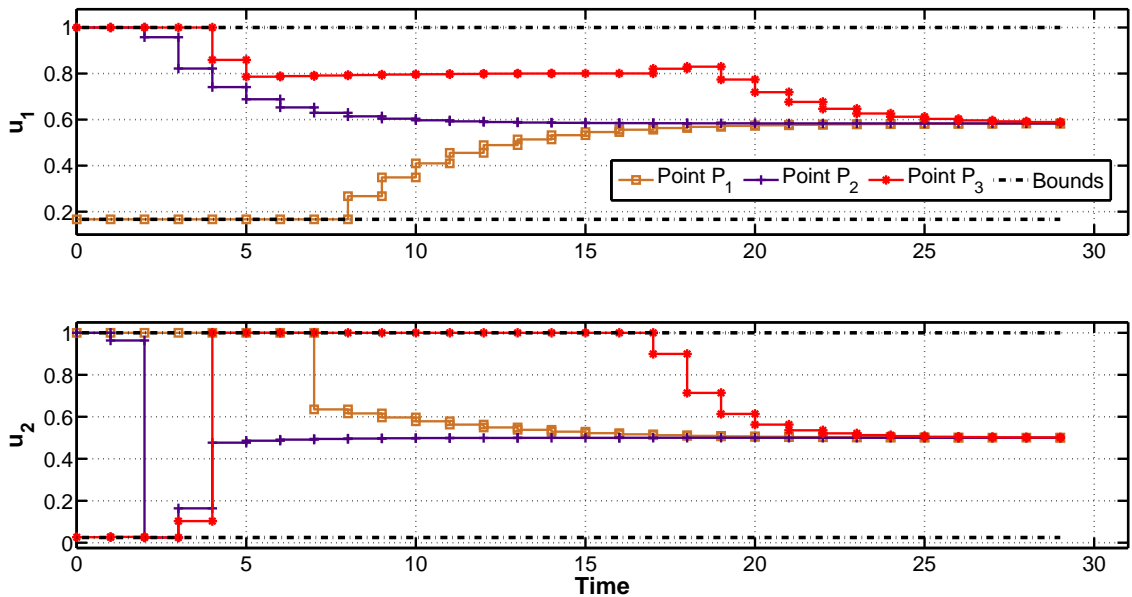}}
\caption{CSTR System: Plot of control inputs} 
\label{Hicks2S_Cont_NMPC_result1_inputs}
\end{figure}

Figure \ref{Hicks2S_Cont_NMPC_result1_xPx} depicts initial condition value i.e. $log_{10} \left[\mathbf{x}(t)^T \mathbf{P} \mathbf{x}(t) \right]$ value along with a limit $log_{10} \alpha$, which represents the terminal set boundary. Initially, values are larger than $\log_{10}\alpha $, which indicates that the initial condition is outside the terminal region. Subsequently, value (in log scale) keeps becoming quiet small indicating that the states converge to the origin i.e. $\mathbf{x}(t) \to \mathbf{0}$ as $t \to \infty$. Logarithmic scale is used because the range of values is higher. 
\begin{figure}[!ht]
\centerline{\includegraphics[width=\columnwidth]{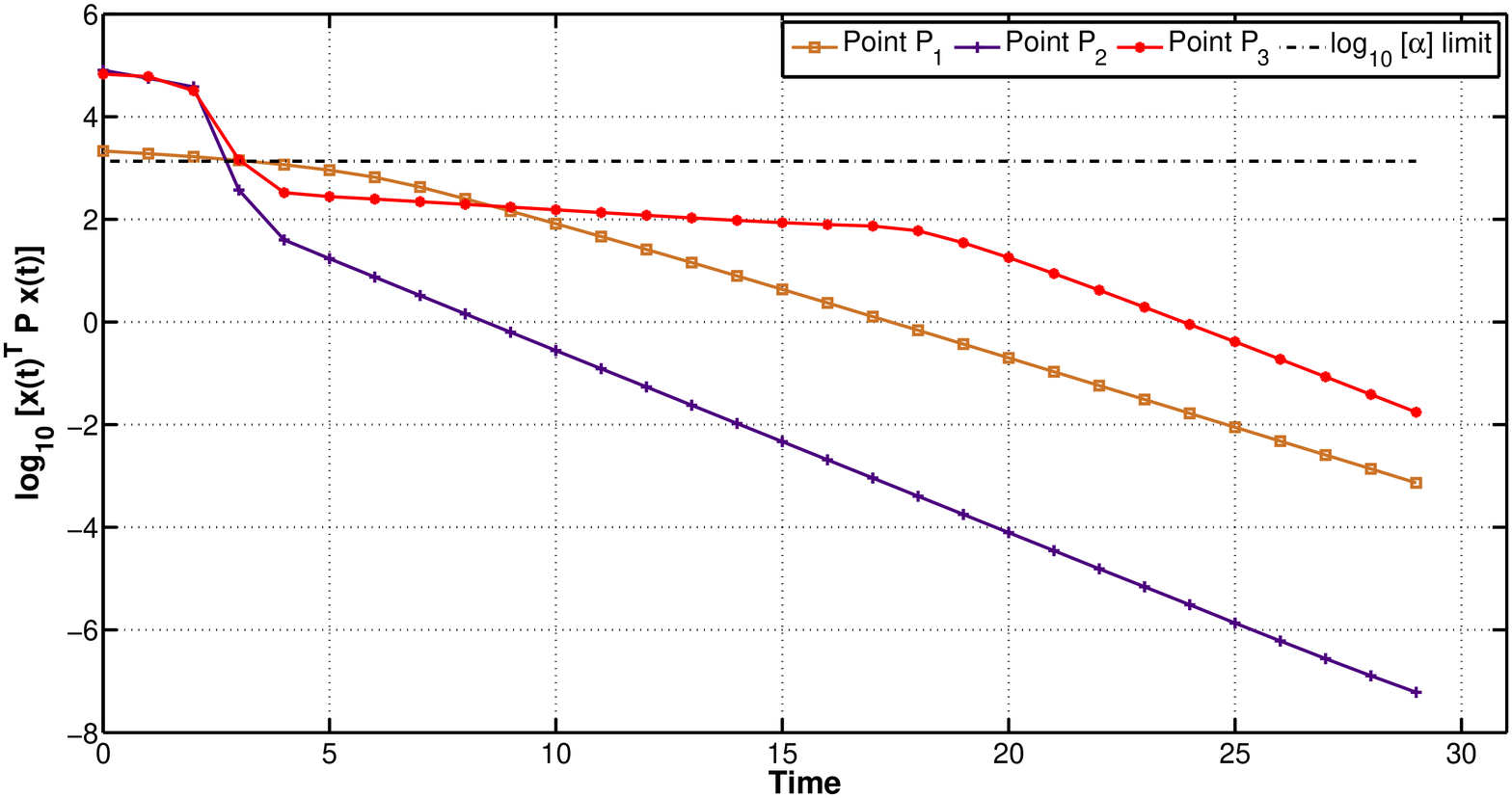}}
\caption{CSTR System: Plot of states in actual variables} 
\label{Hicks2S_Cont_NMPC_result1_xPx}
\end{figure}

Figure \ref{Hicks2S_Cont_NMPC_result1_TCV} shows the terminal constraint value i.e. value of $\mathbf{z}(t+T_p)^{T}\mathbf{P}\mathbf{z}(t+T_p)$ along with its limit $\alpha$. Value of $\alpha$ corresponds to the terminal region boundary. Value always remains below $\alpha$ indicating that the predicted state at the end of the horizon time i.e. $\mathbf{z}(t+T_p)$ is always inside the terminal region i.e. terminal inequality constraint is satisfied every time. 
%Logarithmic function is used because the range of values is higher and it is a monotonically increasing function. 

\begin{figure}[!ht]
\centerline{\includegraphics[width=\columnwidth]{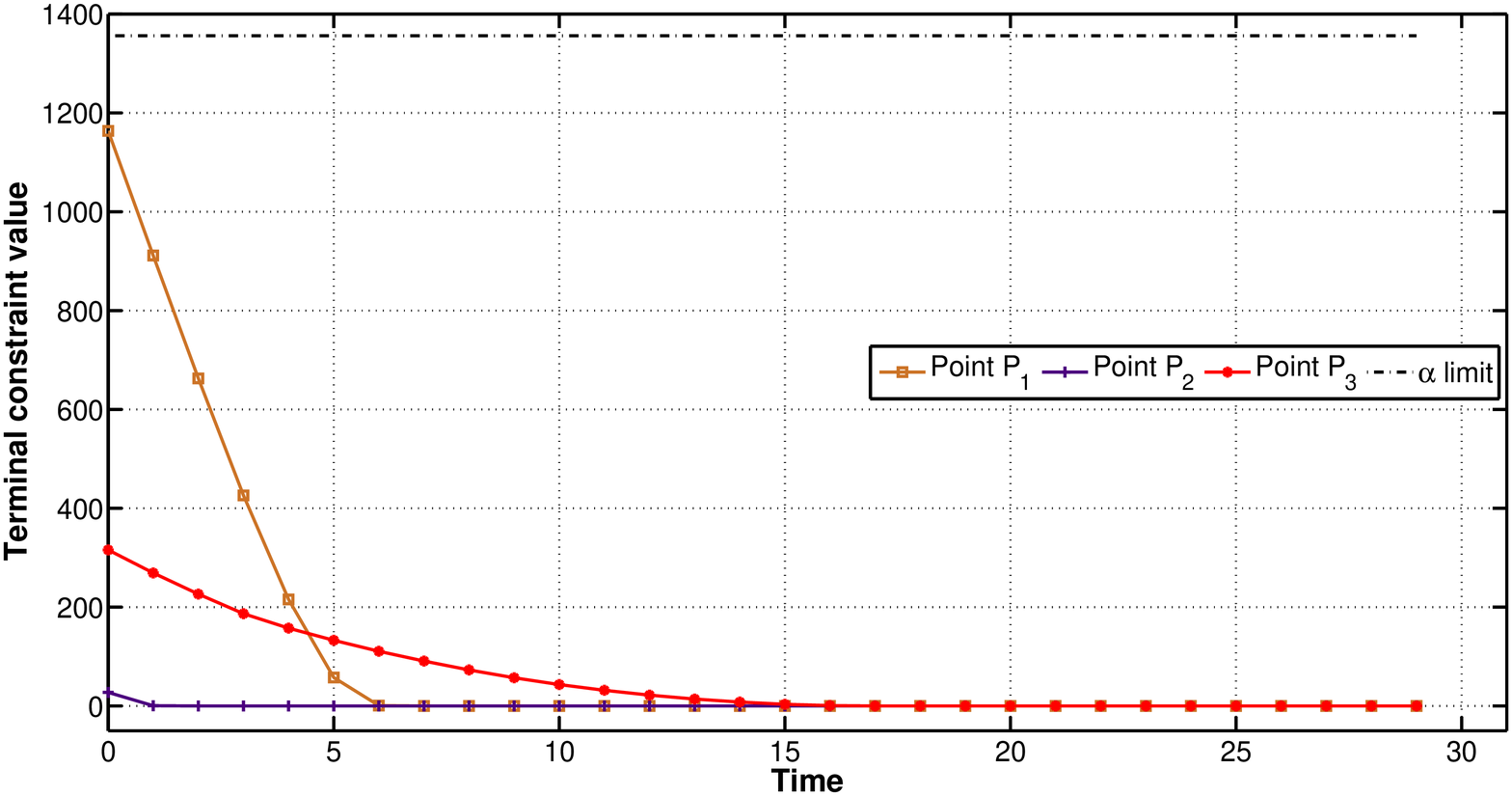}}
\caption{CSTR System: Plot of states in actual variables} 
\label{Hicks2S_Cont_NMPC_result1_TCV}
\end{figure}

Figure \ref{Hicks2S_Cont_NMPC_result1_xx_log} depicts value i.e. $log_{10} \left[\mathbf{x}(t)^T \mathbf{x}(t) \right] = log_{10} |\mathbf{x}(t)|^2$ value. For trajectories starting from the initial conditions $P_3$, value increases slightly at $t=4$, which clearly motivates the need for developing Lyapunov stability theory. However, it can be noted that during the entire trajectory value of the Lyapunov function $\left[\mathbf{x}(t)^T \mathbf{P} \mathbf{x}(t) \right]$ as shown in the figure \ref{Hicks2S_Cont_NMPC_result1_xPx} is continuous decreasing every time. This effectively illustrates the requirement of the presence of the matrix $\mathbf{P}$ in the Lyapunov function. 
% Initially, values are larger than $\log_{10}\alpha $, which indicates that the initial condition is outside the terminal set. Subsequently, value (in log scale) keeps becoming quiet small indicating that the states converge to the origin i.e. $\mathbf{x}(t) \to \mathbf{0}$ as $t \to \infty$. Logarithmic scale is used because the range of values is higher. 

\begin{figure}[!ht]
\centerline{\includegraphics[width=\columnwidth]{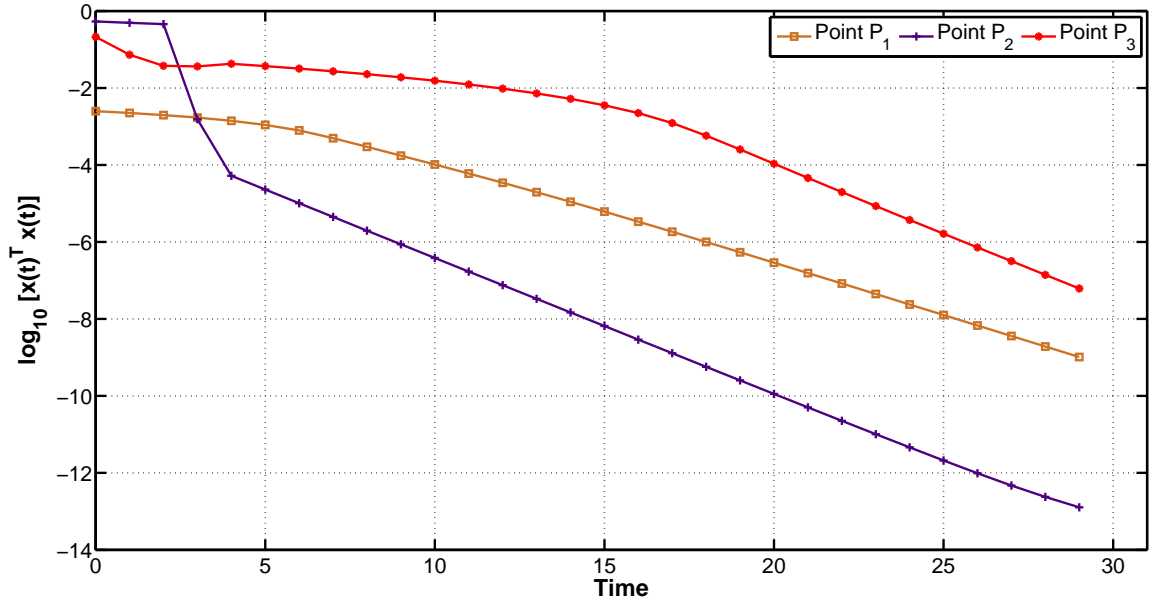}}
\caption{CSTR System: Plot of states in actual variables} 
\label{Hicks2S_Cont_NMPC_result1_xx_log}
\end{figure}

%
%Figure \ref{FourTank30_Cont_NMPC_result1_5s_CT} plots the computation time required for the MPC calculations when terminal set computed using arbitrary controller based approach is used. It can be observed the computation time is sufficiently lesser than the chosen sampling time of $T = 5s$, making the design suitable for practical implementation. 
%\begin{figure}[!ht]
%\centerline{\includegraphics[width=\columnwidth]{Figures/FourTank30_Cont_NMPC_result1_5s_CT.png}}
%\caption{Four tank setup: Plot of states in actual variables} 
%\label{FourTank30_Cont_NMPC_result1_5s_CT}
%\end{figure}

%Computations are carried using an average computer (laptop) with the following configuration: HP 0077tx model (Processor - Intel Core i5 8th generation, RAM - 8 GB, Operating System - Windows 10 Pro 64 bit, MATLAB version - 2013a 64 bit). It may be noted that one can use supercomputers and reduce the computation times to extremely small values, however, novelty of the work lies in reducing the computation time using ordinary computer available in today's time. This implementation makes it practically realizable for everyone using their regular (day to day) computer system without a need of supercomputer. In addition, if one is provided with a high computing machines then the computation times will reduced by similar proportions because of the inherent property of the proposed approach leading to a larger terminal set and smaller prediction horizon time. 

Table \ref{Hicks2S_N} presents approximate minimum prediction horizon time for MPC formulation to be feasible for the chosen initial conditions. It can be noticed that there is significant reduction in the minimum prediction horizon time, which is primarily due to the fact that the size of the terminal regions are larger in the arbitrary controller based approach and LQR based approach when compared to the literature approach. It is well established that the computation time required for MPC optimization convergence reduce exponentially when the prediction horizon time is reduced \cite{Rawlings2017}. Hence, the efficacy of the proposed approaches to significantly reduce the prediction horizon time is effectively demonstrated using the CSTR system case study. 

Since states and inputs in the CSTR case study are converted to dimensionless entities by scaling, the time variable is also scaled. Hence, it would not be legitimate to directly compare the MPC optimization convergence loop time with the sampling time for this case. However, it is observed that the time taken for MPC optimization convergence using literature approaches is significantly larger than the time taken in the case of proposed approaches, which is primarily due to the significantly lesser prediction and control horizon time(s) requirements.

\begin{table}[htbp]
\caption{Minimum prediction horizon time required for feasibility}
\label{Hicks2S_N}
\begin{center}
\begin{tabular}{|l||c|c|c|c|}
\hline 
Approach $\downarrow$ / Point $\rightarrow$ & $P_1$ & $P_2$ & $P_3$ \\ \hline\hline
Chen and Allg\"ower's approach ($\kappa$) \cite{Chen1998} & 15 & 5 & 28  \\ \hline 
Arbitrary controller based approach ($\rho_x, \rho_u$) & 6 & 3 & 11  \\ \hline 
LQR based approach ($\rho_x, \rho_u$) &  4 &  3 &  3  \\ \hline 
\end{tabular}%
\end{center}
\end{table}

\section{Conclusions}
Approaches available in the literature for the terminal region characterization for the continuous time NMPC formulations provide a limited degrees of freedom and often result in a conservative terminal region, thereby resulting in a conservative region of attraction. Larger the terminal region larger is the region of attraction. An arbitrary stabilizing controller based approach and novel LQR based approach is presented in this work which provides a large degrees of freedom for shaping of the terminal region for the continuous time systems. Terminal penalty term is computed using the modified Lyapunov equation and subsequently the nominal asymptotic stability of continuous time NMPC with updated terminal ingredients is established. Proposed approaches provides linear controller gain and two additive matrices as the tuning parameters for enlargement of the terminal region and also makes use of inequality based method.

Efficacy of the both the terminal region characterization approaches is demonstrated using benchmark CSTR system. It is observed that terminal region area obtained using the the arbitrary controller based approach and the novel LQR based approach is approximately 45 and 412 times larger by area as compared to the largest terminal region obtained using Chen and Allg\"ower's inequality based approach from \cite{Chen1998} respectively. Continuous time NMPC simulations validate the asymptotic stability property of the designed controller. It is observed that the minimum prediction horizon required for feasibility of the NMPC formulation using the proposed approaches is significantly smaller than the one required using the literature approach. 

During the simulations for simplicity, tuning parameter matrices are chosen to be multiple of the stage weighting matrices. Future research would involve choosing a completely arbitrary tuning matrices for shaping of the terminal regions. In addition, choosing smaller control horizon time when compared to the prediction horizon time and establishing asymptotic stability is another research direction to explore. 

\bibliography{Journal08_bib1.bib}%
\clearpage

%\section*{Author Biography}
%
%\begin{biography}{\includegraphics[width=66pt,height=86pt,draft]{empty}}{\textbf{Author Name.} This is sample author biography text this is sample author biography text this is sample author biography text this is sample author biography text this is sample author biography text this is sample author biography text this is sample author biography text this is sample author biography text this is sample author biography text this is sample author biography text this is sample author biography text this is sample author biography text this is sample author biography text this is sample author biography text this is sample author biography text this is sample author biography text this is sample author biography text this is sample author biography text this is sample author biography text this is sample author biography text this is sample author biography text.}
%\end{biography}

\end{document}